%% file: main.tex
\documentclass[runningheads]{llncs}
\usepackage{pgf-umlsd}
\usepackage{multicol}
\usepackage{amssymb}
\usepackage{graphicx}
\usepackage{amsmath}
\usepackage{enumerate}
\usepackage{enumitem}
\usepackage{comment}
\usepackage{algorithm}
\usepackage{algpseudocode}
\usepackage{tikz}
\usepackage{booktabs}
\usepackage{ragged2e}
\usepackage{pifont}
\usepackage{xcolor}
\usepackage[none]{hyphenat}
\usepackage{microtype}
\usepackage{float}
\usepackage{mdframed}
\usetikzlibrary{positioning, fit, calc}   
\tikzset{block/.style={draw, thick, text width=8cm ,minimum height=0.75cm, align=center},   
line/.style={-latex}     
}  
\usepackage{marvosym}
\usepackage{hyperref} 
\usepackage{float}
\usepackage{algpseudocode}
\hypersetup{colorlinks=true}
\usepackage{url} 
\usepackage{enumitem}  
\usepackage{multicol}
\usepackage{perpage} 
\MakePerPage{footnote} 

\usepackage{subcaption}
 \usepackage{setspace}

\input{macros.tex}
\begin{document}
\title{Post-Quantum Secure Decentralized Random Number Generation Protocol with Two Rounds of Communication in the Standard Model\thanks{\textit{Manuscript has been accepted for publication in the Proceedings of the 12th International Conference on Future Data and Security Engineering (FDSE 2025, \url{https://thefdse.org/index.html}), published by Springer Verlag in Communications in Computer and Information Science (CCIS) Series, \url{https://www.springer.com/series/7899}.}}
}
\titlerunning{Post-Quantum-DRNG}
%
\author{Pham Nhat Minh\inst{1,2}  \and Khuong Nguyen-An\inst{1,2}\textsuperscript{(\Letter)}
}
\authorrunning{Minh and Khuong}
%
\institute{Department of Computer Science, Faculty of Computer Science and Engineering,\\ 
Ho Chi Minh City University of Technology (HCMUT), 268 Ly Thuong Kiet Street, District 10, Ho Chi Minh City, Vietnam \\
Emails: \email{\{pnminh.sdh232, nakhuong\}@hcmut.edu.vn}\\
\and
Vietnam National University Ho Chi Minh City, Linh Trung Ward,\\ Thu Duc City, Ho Chi Minh City, Vietnam \\
}
%

%

%
\maketitle              
\sloppy 
\begin{abstract} 
Randomness plays a vital role in numerous applications, including simulation, cryptography, distributed systems, and gaming.
Consequently, extensive research has been conducted to generate randomness.
One such method is to design a decentralized random number generator (DRNG), a protocol that enables multiple participants to collaboratively generate random outputs that must be publicly verifiable.
However, existing DRNGs are either not secure against quantum computers or depend on the random oracle model (ROM) to achieve security.
In this paper, we design a DRNG based on lattice-based publicly verifiable secret sharing (PVSS) that is post-quantum secure and proven secure in the standard model.
Additionally, our DRNG requires only two rounds of communication to generate a single (pseudo)random value and can tolerate up to any $t < n/2$ dishonest participants.
To our knowledge, the proposed DRNG construction is the first to achieve all these properties.

\end{abstract}
\section{Introduction and Related Works}
A reliable source of randomness is essential in many applications, including lotteries, gaming, e-voting, simulation, and cryptography. Due to the widespread applications of randomness, substantial efforts have been dedicated to generating random numbers. A natural approach is to rely on a single party, such as \cite{Mads10}, to generate the entire sequence of random numbers. However, this involves the risk that the centralized party might insert trapdoors to bias the supposed random output or secretly sell them to outsiders. Hence, it would be preferable to design protocols that allow \textit{multiple participants} to participate together in the random generation process to reduce the role of a trusted third party. In addition, these protocols must be publicly verifiable by an external verifier. This protocol is called \textit{ a decentralized random number generator} (DRNG) \cite{nguyen2019system,SSTE20}. With the development of decentralized applications, the number of DRNGs has begun to rise very rapidly \cite{SJKGGKFF17,CD17,GHMVZ17,SJSW18,DGKR18,nguyen2019system,CSS19,CD20,GLOW20,DFI22,DKMR21,BLLOP23,CATB23,ZLOKZ25}. Unfortunately, the majority of DRNGs depend on the hardness of either the \textit{discrete log} (DL) or the integer factorization problem for security. These problems can be solved by the Shor quantum algorithm \cite{Schor99}. Hence, an adversary can use a quantum computer to predict the outputs of the DRNG, compromising the security of these protocols.

Another problem with these DRNGs is that most existing DRNGs have been proven to secure the random oracle model (ROM) instead of the standard model. Technically, ROM assumes that all participants have access to a \textit{public} oracle that i) returns a truly random output for each unqueried input, and ii) if the input were queried previously, it would return the same previous output \cite{BR93,CGH04}. Such an oracle cannot be implemented in the real world \cite{CGH04}, so assuming such an oracle is a nonstandard assumption. The best one can do in practice is to heuristically instantiate the oracle with a specific cryptographic hash function. However, there have been \textit{counterexamples} of protocols that are secure in the ROM, but become insecure when the random oracle (RO) is instantiated with any hash function, such as \cite{GK03,CGH04,KRS25}.
These counterexamples raise the concern that most natural protocols secure in ROM could still achieve security in the real world. Among DRNGs, except for several PVSS-based constructions such as \cite{CD17,DKMR21} (which are not post-quantum secure), existing DRNGs need to employ ROM to achieve security. The use of ROM is due to one of the following: i) their protocols rely on non-interactive zero-knowledge arguments (NIZK) and employ the Fiat-Shamir paradigm (such as \cite{KRDO17,CD20}), or ii) the current output is computed as the hash digest of the partial outputs combined (such as \cite{GHMVZ17,SJSW18,DGKR18,SJHSW20}), and this hash must be modeled as some RO. In either case, the use of ROM is not a standard method to capture security in the real world.

Another desired property in any protocol is to minimize the number of rounds (round complexity). Indeed, in a network, sending a message to others requires an amount of time due to network latency time $\Delta$. So if there are $r$ communication rounds, it takes at least $r\cdot \Delta$ time to complete the protocol. Also, if many online rounds are required, participants need to be online and responsive more frequently, and are more likely to be vulnerable to network-based attacks. Hence, it would be essential to design protocols to minimize the number of rounds so that i) the time wasted by network latency is reduced, ii) participants can act more independently and are less likely to be affected by network-related attacks.  If only the standard model and post-quantum security are required, then one might attempt to construct DRNGs from coin-flipping protocols \cite{BD84,B85}  where, in each execution, the participants agree on one single random bit. However, to generate a random number of $\lambda$ bits, participants need at least $\Omega(\lambda)$ rounds of communication. In addition, coin-flipping protocols allow adversaries to bias the output to an extent, depending on the round complexity \cite{DGLD22,C86} or number of participants \cite{BD84} for generating a single bit. Hence, generating random outputs by coin-flipping protocols will require a lot of rounds, which is not desirable.

Given the importance of the three properties above, it would be desirable to construct a DRNG that  \textit{can achieve post-quantum security} in \textit{the standard model}, and the round complexity should \textit{be as small as possible}.
\subsection{Our Contribution}\label{section-our-contribution}
 We construct a DRNG that is i) \textit{post-quantum secure}, ii) requires only \textit{$2$ rounds of communication} to generate a single random value, iii) achieves security in the \textit{standard model}, and finally, iv) can tolerate any $t<n/2$ dishonest participants using the lattice-based PVSS of \cite{MNSN25} (The PVSS is also post-quantum secure, and is proven secure in the standard model. In fact, these properties in the PVSS imply the desired properties in our DRNG). As a trade-off, our construction requires all participants to agree on a common reference string (CRS, which can be understood as a string drawn from a specific distribution and is broadcast to all participants) generated by a third party. However, this is acceptable, as we need to rely on a third party once to generate a correct CRS. Then participants use the string to jointly generate many random numbers (without further third-party involvement), instead of letting a third party control the whole random generating process (so its role is \textit{greatly reduced}, but not entirely removed). In fact, having a CRS is also implicitly required for other DRNGs as well: For example, in DRNGs rely on the DL problem, participants have to agree on a cyclic group $\mathbb{G}$ with generator $g$ where the DL problem is assumed to be hard (the CRS). Such a group needs to be proposed by a third party (such as the parameter of Secp256k1 was recommended in \cite{Sec2}). For more examples, \cite{CATB23} also relies on a third party to generate the CRS for participants (but only achieves security in ROM). Now, to sum up, assuming participants having access to a CRS, our DRNG is the first to simultaneously satisfy all four properties: i) secure against quantum computers, ii) require only two rounds, iii) proven secure the standard model, and iv) can tolerate up to $t<n/2$ dishonest participants. 

 \subsection{Related Works}
 \noindent \textbf{Existing DRNGs.} Various primitives can be taken to construct DRNGs as follows: i) DRNG from hash, ii) DRNG from PVSS, iii) DRNG from verifiable random function (VRF) \cite{MRV99}, iv) DRNG from verifiable delay function (VDF) \cite{BBBF19}, and v) DRNG from homomorphic encryption (HE). 
 
 Hash-based DRNGs, such as \cite{Ra17,AMM18}, are insecure. For example, in RANDAO \cite{Ra17}, each participant $P_i$ generates a secret $s_i$ and needs to provide a hash digest of $s_i$ as its commitment. After all participants have provided their digests, the participants $P_i$ reveal the corresponding secret $s_i$, and the output $\Omega$ is calculated as the XOR of all secrets $s_i$. However, the last participant can choose not to reveal his $s_i$ if $\Omega$ does not benefit him, causing the protocol to abort and restart (also known as \textit{withholding attack}). Hence, the output of these protocols can be biased. Some VRF-based protocols, such as \cite{GHMVZ17,DGKR18}, also have the same issue.
 
 DRNGs from PVSS can be divided into two approaches: With leaders and without leaders. Those with leaders such as \cite{SJSW18,BSLKN21,ZLOKZ25} only achieve weak security in the sense that the unpredictability of the output is only achieved from the $(t+1)$-th epoch. PVSS-based DRNGs without leaders \cite{CD17,SJKGGKFF17,CD20} or threshold VRFs \cite{GLOW20,DFI22,BLLOP23} enjoy full security properties. Unfortunately, they are based on the discrete log assumption. Hence, they are not secure against quantum computers.

For VDF-based DRNG, we first recall VDF: It is an algorithm that requires a lot of time to compute the outputs, but it can be verified very quickly. One idea for constructing DRNG is to combine the VDF with RANDAO: After participants execute RANDAO in a short period of time, they receive a value $\Omega'$, which is the input of the VDF to generate the final output $\Omega$. Since it takes a very long time to compute $\Omega$ from $\Omega'$, the adversary is not incentivized to keep his secret, or he will be discarded without knowing whether $\Omega$ benefits him or not. However, it does not prevent an adversary from learning the output sooner than the participants. Adversaries with improved hardware might compute the VDF output several times faster \cite{Vi18} (say, $5$ times \cite{HHY20}). Because VDF requires a very long time to compute (say, a day), $5$ times faster means that the adversary will know the output much sooner (several hours) and can secretly leak it for their benefit. Another idea is to use VDF with trapdoors \cite{SJHSW20}. In trapdoor VDF-based DRNGs such as \cite{SJHSW20,CATB23}, the idea is that each participant has a VDF with a trapdoor, and the participants will use the trapdoor to quickly compute the outputs. Then, other participants will compute the output of any participant who withholds it, but they need a lot of time to do so. The problem is that if withholding happens, honest participants must compute the VDF result without trapdoors, and hence, they still know the output much later than dishonest ones.

 The construction using HE includes \cite{nguyen2019system,Ng19} and \cite{CSS19}. In \cite{nguyen2019system,Ng19}, the participants collaborate with the client to generate randomness. The client generates a pair of public-secret key pairs $(\pk,\sk)$, and then the participants use $\pk$ to encrypt their contributions. The client then aggregates the encryption and uses $\sk$ to obtain randomness. However, the scheme requires every client to be honest and it is pointed out that the scheme becomes insecure when any dishonest client colludes with a dishonest participant to manipulate the output. The scheme of \cite{CSS19} is based on a threshold decryption scheme and achieves full security properties. However, it is not post-quantum secure as it relies on the discrete log problem (and the construction also relies on the ROM to achieve security). 
 
To our knowledge, no work has so far constructed a post-quantum secure DRNG with at most $2$ rounds in the standard model, and our DRNG is the first one to do so. For a comparison with selected previous works, we refer to Table \ref{table:comparision of DRNG}. \\
 
\noindent \textbf{HashRand.} HashRand \cite{BBBK24} also attempted to propose a post-quantum DRNG from verifiable secret sharing. Compared to theirs, our protocol incurs higher communication and computation complexities (more time \textit{incurred by the local computation of participants}). However, our protocol requires only two rounds to generate a random output, and participants are guaranteed to agree on a random value. HashRand, on the other hand, requires at least three rounds (using Gather primitive in \cite{AJMMST21}), can only tolerate $t<n/3$ dishonest participants, and there is a probability $p$ that participants can disagree on the random output. As in \cite[Section 3.3]{BBBK24}, the lower $p$ is, the more rounds HashRand requires. Also, HashRand is hash-based and relies on ROM, while ours is based on lattice assumption in the standard model. Hence, our construction has advantages over Hashrand in that it requires fewer rounds (less time \textit{incurred by network latency}), tolerates more dishonest participants, and achieves security in the standard model.

\begin{table}\centering
\caption{\small{Comparison with selected DRNGs. Communication complexity refers to the total number of bits sent by participants. For computation complexity/node, we mean the number of computation steps by a single participant. This is one of the two factors affecting the total time to generate an output. The other is the number of rounds, which also affects the total time, since more rounds imply more time wasted by network latency, as discussed earlier. In the table, we consider $n$ as the main parameter in these complexities, similar to other works \cite{CSS19}.  Also, $p$ is the probability that participants agree on the same random output in HashRand.  We use $\Omega$ for estimating the costs of our DRNG. The reason for this is described in Section \ref{seciton-complexity-analysis}. }}
{\scalebox{0.7}{
\begin{tabular}{lccccccccccc}\label{table:comparision of DRNG}
\\
\hline
\hline
             & \rotatebox{270}{\textit{Pseudo-randomness}} & \rotatebox{270}{\textit{Unpredictability}} & \rotatebox{270}{\textit{Bias-resistance }} & \rotatebox{270}{\textit{Liveness}} & \rotatebox{270}{\textit{Public Verifiability}} & \rotatebox{270}{\textit{Comm. Cost}} & \rotatebox{270}{\textit{Comp. Cost/ node}} 
             & \rotatebox{270}{\textit{Rounds/value\quad\quad}}
             & \rotatebox{270}{\textit{Honest nodes}} & \rotatebox{270}{\textit{Primitives}} & \rotatebox{270}{\textit{Assumptions}} \\
\hline
\hline
            RANDAO \cite{Ra17} & \xmark & \checkmark & \xmark & \xmark & \checkmark & 
              $O(n)$ & $O(n)$ & $2$ & $1$ & Com.+ Rev. & Hash\\
             Scrape \cite{CD17} & \checkmark & \checkmark & \checkmark & \checkmark  & \checkmark 
             & $O(n^2)$ & $O(n^2)$ & $2$ & ${n}/{2}$ & PVSS & DLOG+CRS\\
              HydRand \cite{SJSW18} & \checkmark & \checkmark & \checkmark & \checkmark &    \checkmark &
              $O(n^2)$ & $O(n)$ & $2$ & $2n/3$ & PVSS & DLOG+CRS+ROM\\
                Algorand \cite{GHMVZ17} & \xmark & \checkmark &  \xmark & \checkmark &   \checkmark &
              $O(n)$ & $O(n)$ & $1$ & ${2n}/{3}$ & VRF & CRS+ROM\\

              Nguyen \etal \cite{nguyen2019system} & \xmark & \checkmark &  \xmark & \xmark &   \checkmark &
              $O(n)$ & $O(n)$ & $1$ & $1$ & HE& DLOG+CRS+ROM\\
              Randrunner \cite{SJHSW20} & \xmark & \checkmark &  \checkmark & \checkmark & \checkmark &  $O(n)$ & VDF time & $1$ & ${n}/{2}$ & VDF & Factoring+ROM\\

  Albatross \cite{CD20} & \checkmark & \checkmark & \checkmark & \checkmark  & \checkmark 
             & $O(n^2)$ & $O(n^2 \log^2 n)$ & $2$ & ${n}/{2}$ & PVSS & DLOG+CRS+ROM\\

              drand \cite{Dr20} & \checkmark & \checkmark &  \checkmark & \checkmark 
             & \checkmark & $O(n)$ &  $O(n\log^2n)$  & $1$ & ${n}/{2}$ & TSS.& DLOG+CRS+ROM\\

             Bicorn \cite{CATB23} & \xmark & \checkmark & \checkmark & \checkmark & \checkmark & 
              $O(n)$ & VDF time & $2$ & $2$ & VDF & CRS+Factoring+ROM \\

                  HashRand \cite{BBBK24} & \xmark & \checkmark &  \checkmark & \checkmark 
             & \checkmark & $O(n^2 \log^2 n)$ &  $O(n^2\log^2 n)$  & $\log 1/p$ & ${n}/{3}$ & VSS & Hash (needs ROM)\\
\hline
              {\textbf{Ours}} & \checkmark & \checkmark &  \checkmark & \checkmark &\checkmark
             & $\Omega(n^2)$ &  $\Omega(n^3)$  & $2$ & ${n}/{2}$ & PVSS & LWE+CRS\\
\hline
\hline
\end{tabular}
}}
\end{table}

 \subsection{Why from PVSS}
 One might ask why we use PVSS to construct the DRNG instead of a distributed VRF such as \cite{DFI22,GLOW20}. The reason is that we aim to achieve post-quantum security in the standard model. However, current techniques for lattice-based VRF, such as \cite{ESLR22}, require ROM to achieve security. Unfortunately, we are not aware of any direction in which to construct lattice-based VRFs in the standard model. Howerver, as shown by \cite{MNSN25}, it is possible to construct lattice-based PVSS in the standard model, and as a result, the PVSS allows achieving post-quantum DRNG in the standard model. Therefore, we decided to construct DRNG from the lattice-based PVSS of \cite{MNSN25} instead.

\subsection{Paper Organization}
The rest of the paper is organized as follows: Section \ref{section-preliminaries} recalls the necessary background. Section \ref{section-the-underlying-pvss} briefly describes the PVSS of \cite{MNSN25}, which is the lattice-based PVSS we will use in our DRNG. Section \ref{section-drng} describes the generic DRNG construction from any PVSS without RO. We also provide the security proof of the generic DRNG and analyze the complexities when instantiated with the PVSS of \cite{MNSN25}. Finally, Section \ref{section-conclusion} concludes our paper.

\section{Preliminaries}\label{section-preliminaries}

 For $p \geq 2$, denote $\mathbb{Z}_p$ to be the ring of integers mod $p$. For $\bx \in \mathbb{Z}^v$, let $\rho_{\sigma}(\bx)=e^{-\pi \cdot ||\bx||^2/\sigma^2}$, where $||\bx||$ is the Euclidean norm of $\bx$. We denote $D_{\mathbb{Z}^v,\sigma}$ be the discrete Gaussian distribution that assigns probability equal to $\rho_{\sigma}(\bx)/(\sum_{\by \in \mathbb{Z}^v} \rho_{\sigma}(\by))$ for each $\bx \in \mathbb{Z}^v$. We denote $[n]=\{1,\dots,n\}$ and $\negl(\lambda)$ to denote a negligible function in $\lambda$ (see \cite{ACPS09,CCHLRRW19} for definition). We denote $x \leftarrow \mathcal{D}$ to say that $x$ is sampled from a distribution $ \mathcal{D}$ and  $x \xleftarrow{\$}
 \mathcal{S}$ to say that $x$ is uniformly sampled from a set $\mathcal{S}$. We consider a \textit{synchronous} network: The time between messages is bounded within a value $\Delta$. We assume a \textit{broadcast channel}, where everyone can see a message broadcasted by a participant. The adversary $\Adversary$ is \textit{static}: $\Adversary$ corrupts a set of $t<n/2$ participants before the start of the protocol and acts on their behalf. In the DRNG, we use the word \textit{epoch}. The $r$-th epoch is the $r$-th time participants execute the DRNG to generate an output. Each epoch is divided into rounds. In each round, participants perform some local computation and then simultaneously \textit{exchange messages} (in parallel) with each other, before going to the next round. Some protocols \cite{SJSW18,GLOW20} call an epoch a round, but we use the definition above to analyze the round complexity in a single execution (similar to other works  \cite{CGGSP21,K24}).  

\subsection{Lattices}\label{section-lattice}
\begin{definition}[LWE Assumption \cite{Re09}]\label{definition-lwe}
Let $u,v,q$ be positive integers, and let $\alpha$ be a positive real number. Let $\bs \in \mathbb{Z}^v$ be drawn from some distribution. For any PPT adversary $\Adversary$, there is a negligible function $\negl$ such that

$$\left| \condprob{b=b'}{\bA\uniform \mathbb{Z}_q^{v\times u},\be \leftarrow D_{\mathbb{Z}^u,\alpha q}, b \uniform \{0,1\}, \\
\text{If $b=0$, $\bb= \bs^\top\bA+\be^\top$, else $\bb\uniform \mathbb{Z}_q^u$,}\\
b' \leftarrow \Adversary(\bA,\bb)\\
}-\dfrac{1}{2} \right| \leq \negl(\lambda).$$
\end{definition}
It is proved that when $\bs \uniformly \mathbb{Z}_q$ or $\bs \leftarrow D_{\mathbb{Z}^v,\alpha q}$, then breaking LWE is known to be hard for quantum computers as long as $1/\alpha$ is sub-exponential in $v$.

\subsection{Non-Interactive Zero-Knowledge Arguments}\label{section-nizk}
A non-interactive zero-knowledge argument (NIZK) allows a prover to non-interactively prove that it knows a witness to a valid statement without revealing any information about the witness. Here we recall the syntax of NIZK. 
\begin{definition}[NIZK, Adapted from \cite{LNPT20,MNSN25}]
Let $\LL=(\LL_{zk},\LL_{sound})$ be a gap language with corresponding relation $\RR=(\RR_{zk},\RR_{sound})$ and let $\CRS$ be a set of common reference string. A NIZK argument for $\LL$ with a common reference string set $\CRS$ is a tuple $\NIZK=(\NIZK.\Setup,\NIZK.\Prove,\NIZK.\Verify)$ as follows:
\begin{itemize}
\item $\NIZK.\Setup(1^\lambda,\LL) \rightarrow \crs:$ This is an algorithm executed by a third party. This PPT algorithm outputs a common reference string $\crs \in \CRS$.
\item $\NIZK.\Prove(\crs,x,w) \rightarrow \pi:$ This is an algorithm executed by the prover to prove that it knows a witness $w$ corresponding to a statement $x \in \LL_{zk}$ s.t. $(x,w) \in \RR_{zk}$. It outputs a proof $\pi$ certifying $(x,w) \in \RR_{zk}$.
\item $\NIZK.\Verify(\crs,x,\pi) \rightarrow 0/1:$ This is an algorithm executed by the verifier. It outputs a bit $b \in \{0,1\}$ which certifies the validity of $(x,\pi)$.
\end{itemize}
\end{definition}

We require an NIZK to satisfy three properties: \textit{correctness, adaptive soundness} and \textit{multi-theorem zero-knowledge}, where the language $\LL_{sound}$ is used in the adaptive soundness property. However, due to limited space, we do not describe these properties here because we do not use them in this paper. The properties are only used in \cite{MNSN25} for the security of the PVSS there. In this paper, we only require the syntax of the NIZK to describe the PVSS of \cite{MNSN25} in Section \ref{section-pvss}. Thus, for the definition of these properties, we refer to \cite{CCHLRRW19,LNPT20,MNSN25}.

\subsection{Publicly Verifiable Secret Sharing}\label{section-pvss-def}
We recall the definition of PVSS and its security properties.

\begin{definition}[PVSS, Adapted from \cite{MNSN25}]
    A $(n,t)$-PVSS with $0 \leq t < n/2$ is a tuple of algorithms $\PVSS=(\PVSS.\Setup,$ $\PVSS.\Key\Gen, \PVSS.\Key\Verify,$ $\PVSS.\Share,$ $\PVSS.\Share\Verify,\PVSS.\Dec\Verify,$ $\PVSS.\Combine)$, specified as below. 
    \begin{itemize}
    \item $\PVSS.\Setup(1^\lambda) \rightarrow \pp:$ This algorithm is run by a trusted third party. On input the security parameter $\lambda$, it returns a public parameter $\pp$.
    \item $\PVSS.\Key\Gen(\pp) \rightarrow ((\pk,\sk),\pi):$ This algorithm is run by each participant. It returns a public-secret key pair $(\pk,\sk)$ and a proof $\pi$ of valid key generation.
    \item $\PVSS.\Key\Verify(\pp,\pk,\pi) \rightarrow 0/1:$ This algorithm is run by a public verifier; it outputs a bit $0$ or $1$ certifying the validity of the public key $\pk$.
    \item $\PVSS.\Share(\pp,(\pk_i)_{i=1}^{n'},s,n',t) \rightarrow (E=(E_i)_{i=1}^{n'},\pi):$ This algorithm is executed by the dealer to share the secret $s$ for $n' \leq n$ participants. It outputs the ``encrypted shares'' $E=(E_i)_{i=1}^{n'}$ and outputs a proof $\pi$ for correct sharing.
    \item $\PVSS.\Share\Verify(\pp,(\pk_i)_{i=1}^{n'},n',t,E,\pi) \rightarrow 0/1:$ This algorithm is run by a verifier,  it outputs a bit $0$ or $1$ certifying the validity of the sharing process.
    \item $\PVSS.\Dec(\pp,\pk_i,\sk_i,E_i) \rightarrow (s_i,\pi):$ This algorithm is executed by participant $P_i$. It generates a decrypted share $s_i$ and a proof $\pi$ of correct decryption.
    \item $\PVSS.\Dec\Verify(\pp,(\pk_i,E_i,s_i),\pi):$ This algorithm is run by a public verifier, it outputs a bit $0$ or $1$ certifying the validity of the decryption process.
    \item $\PVSS.\Combine(\pp,S,(s_i)_{i \in S}) \rightarrow s/\perp:$ This algorithm is executed by a public verifier. For a set $S$ and a tuple of shares $(s_i)_{i \in S}$, it outputs the original share $s$ or $\perp$ if the secret cannot be reconstructed.
   
    \end{itemize}
\end{definition}
We require PVSS to satisfy \textit{correctness, verifiability, and IND2-privacy}. For correctness, if an honest dealer shares $s$, then it will output $s$ in the reconstruction phase. For verifiability, if the dealer and all participants have passed verification to share a secret $s$, then the sharing and reconstruction process must be done correctly. For IND2-privacy, we require that for any secrets $s^0,s^1$, it is infeasible for an adversary to distinguish between the transcript of sharing of $s^0$ and $s^1$.

\begin{definition}[Correctness \cite{MNSN25}]\label{def-correctness}
We says that $\PVSS$ achieves \textup{correctness} if for any PPT adversary $\Adversary$, the game $\GAME^{\PVSS-\Correctness}(\Adversary)$ in Figure \ref{figure-game-correctness} outputs $1$ with probability $1-\negl(\lambda)$ for some negligible function $\negl$. 

\begin{framedfigure}[Game $\GAME^{\PVSS-\Correctness}(\Adversary)$ \label{figure-game-correctness}] 
\vspace{0.05cm}
\footnotesize $\pp\leftarrow\PVSS.\Setup(1^\lambda)$, \footnotesize $\CC\leftarrow\Adversary(\pp)$. If $|\CC|>t$ return $0$.\\
\footnotesize $((\pk_i,\sk_i),\pi^0_i) \leftarrow \PVSS.\Key\Gen(\pp)~\forall~i\not \in \CC$, $(\pk_i,\pi_i^0)_{i \in \CC}\leftarrow\Adversary(\pp,\CC)$,\\
\footnotesize \textcolor{gray}{Correctness of key generation.}\\
\footnotesize If $\exists ~i \not \in \CC$ s.t. $\PVSS.\Key\Verify(\pp,\pk_i,\pi^0_i)=0$, return $0$.\\
\footnotesize Let $G=\{i \in [n]~|~\PVSS.\Key\Verify(\pp,\pk_i,\pi^0_i)=1\}$. \footnotesize Assume that $G=[n']$ and $[n] \setminus \CC \subseteq G$. If not, we re-enumerate the participants $P_i$ with $i \in G$ with an element in $[n']$.\\
\footnotesize $s\uniformly \mathbb{Z}_p,(E=(E_i)_{i=1}^{n'},\pi^1) \leftarrow \PVSS.\Share(\pp,(\pk_i)_{i=1}^{n'},s,n',t)$.\\ 
\footnotesize \textcolor{gray}{Correctness of sharing.}\\
\footnotesize If $\PVSS.\Share\Verify(\pp,(\pk_i)_{i=1}^{n'},n',t,E,\pi^1)=0$, return $0$.\\

\footnotesize $(s_i,\pi^2_i) \leftarrow \PVSS.\Dec(\pp,\pk_i,\sk_i,E_i)~\forall~i \not \in \CC$,\\ 
\footnotesize $(s_i,\pi^2_i)_{i  \in [n']\cap \CC}\leftarrow \Adversary(\pp,(\pk_i,\pi^0_i)_{i \in [n']},E,\pi^1,(s_i,\pi^2_i)_{i \not \in \CC})$.\\
\footnotesize \textcolor{gray}{Correctness of share decryption.}\\
If $\exists~i \not \in \CC$ s.t. $\PVSS.\Dec\Verify(\pp,\pk_i,E_i,s_i,\pi^2_i)=0$, \footnotesize return $0$.\\
\footnotesize \textcolor{gray}{Correctness of reconstruction. Any $t+1$ participants who passed $\PVSS.\Dec\Verify$ must agree on $s$.}\\
\footnotesize Let $S=\{i~|~\PVSS.\Dec\Verify(\pp,E_i,s_i,\pi^2_i)=1\}$. If $|S| < t+1$, return $0$. If there exists some $S' \subseteq S$, $|S'| \geq t+1$  such that $\PVSS.\Combine(\pp,S',(s_i)_{i \in S'}) \neq s$, return $0$.\\
\footnotesize Return 1.

\end{framedfigure}
\end{definition}

For verifiability, we need that i) If $(E,\pi)$ is accepted by $\PVSS.\Share\Verify$, then after honestly decrypting $E_i$ from $\PVSS.\Dec$ to obtain $s_i$, it holds that $(s_1,s_2,\dots,s_n)$ are valid shares of some secret $s$ and ii) If participant $P_i$ submits $s'_i$ that causes $\PVSS.\Dec\Verify$ to accept, then $s'_i=s_i$. In this way, from a valid transcript $(E,\pi)$, participants would agree on a unique $s$ in the reconstruction phase, and a verifier will be convinced that both the sharing and reconstruction phases for $s$ are done correctly. Due to limited space, we only recall the definition of valid share language and verifiability. For more details, we refer to \cite{MNSN25}. 

\begin{definition}[Valid Share Language \cite{MNSN25}]\label{definition-well-formedness}
We say that $\LL^\Share_{t} \subseteq \bigcup_{i=t+1}^n \mathbb{Z}_p^i$ is a \textup{valid share language} if:  For any $n' \leq n$ and $(s_1,s_2,\dots,s_{n'}) \in \LL^{\Share}_t$, there is a value $s \in \mathbb{Z}_p$ such that for any $S \subseteq[n']$ with $|S| \geq t+1$, it holds that $\PVSS.\Combine(\pp,S,(s_i)_{i \in S})=s$.
\end{definition}

\begin{definition}[Verifiability \cite{MNSN25}] \label{definition-verifiability}
 We say $\PVSS$ achieves \textup{$(\LL^\Key,\LL^\Share_{t})$ -verifiability} if i) Each instance in $\LL^\Key$ has a unique witness, ii) $\LL^\Share_{t}$ is a valid share language and iii) if for any PPT adversary $\Adversary$, the game $\GAME^{\PVSS-\Verify}(\Adversary)$ in Figure \ref{figure-game-verifiability} outputs $1$ with negligible probability $\negl(\lambda)$.

\begin{framedfigure}[Game $\GAME^{\PVSS-\Verify}(\Adversary)$ \label{figure-game-verifiability}] 
\vspace{0.1cm}
\footnotesize $\pp\leftarrow\PVSS.\Setup(1^\lambda)$. Parse $\pp=(\pp',\pp^\star)$.  $\CC\leftarrow\Adversary(\pp)$. \footnotesize If $|\CC|>t$ return $0$.\\
\footnotesize $((\pk_i,\sk_i),\pi^0_i) \leftarrow \PVSS.\Key\Gen(\pp)~\forall~i\not \in G\cap \CC$, $(\pk_i,\pi_i^0)_{i \in \CC}\leftarrow\Adversary(\pp,\CC)$,\\
\footnotesize Let $G=\{i \in [n]~|~\PVSS.\Key\Verify(\pp,\pk_i,\pi^0_i)=1\}$. \footnotesize Assume that $G=[n']$ and $[n] \setminus \CC \subseteq G$. If not, we re-enumerate the participants $P_i$ with $i \in G$ with an element in $[n']$.\\

\footnotesize $(E=(E_i)_{i=1}^{n'},\pi^1)\leftarrow \Adversary(\pp,(\pk_i,\pi^0_i)_{i \in G})$.\\

\footnotesize $(s_i,\pi^2_i) \leftarrow \PVSS.\Dec(\pp,\pk_i,\sk_i,E_i)~\forall~i \not \in \CC$,\\ 
\footnotesize $(s'_i,\pi^2_i)_{i  \in G\cap \CC}\leftarrow \Adversary(\pp,(\pk_i,\pi^0_i)_{i \in [n]},E,\pi^1,(s_i,\pi^2_i)_{i \not \in \CC})$.\\
\footnotesize \textcolor{gray}{Verifiability of key generation.}\\ If $(\pp',\pk_i) \not \in \LL^\Key$ for some $i \in G \cap \CC$, return $1$.\\
\footnotesize \textcolor{gray}{Verifiability of sharing.}\\
\footnotesize At this point, consider unique $(\sk_i)_{i \in G\cap \CC}$ s.t. $((\pp',\pk_i),\sk_i) \in \RR^\Key~\forall i \in G\cap \CC$.\\
\footnotesize Let $(s_i,.)\leftarrow \PVSS.\Dec(\pp,\pk_i,\sk_i,E_i)~\forall~i \in G \cap \CC$. If $(s_1,s_2,\dots,s_{n'}) \not \in \LL^\Share_{t}$ and $\PVSS.\Share\Verify(\pp,(\pk_i)_{i=1}^{n'},n',t,E,\pi^1)=1$, return $1$.\\
\footnotesize \textcolor{gray}{Verifiability of decryption.}\\
\footnotesize If $s'_i \neq s_i$ and $\PVSS.\Dec\Verify(\pp,\pk_i,E_i,s'_i,\pi^2_i)=1$ for some $i \in G\cap \CC$, return $1$.\\
\footnotesize If $\PVSS.\Dec\Verify(\pp,\pk_i,E_i,s_i,\pi^2_i)=0$ for some $i \in G\not \in \CC$, return $1$.\\
\footnotesize Return $0$.

\end{framedfigure}
\end{definition}

\begin{definition}[IND2-Privacy \cite{MNSN25}]\label{definition-pvss-privacy}
We say that $\PVSS$ achieves \textup{IND2-privacy} if for any PPT adversary $\Adversary$, then there is a negligible function $\negl$ s.t. $\Adv^{\PVSS-\mathsf{IND}}(\Adversary)=|\Pr[\GAME^{\PVSS-\mathsf{IND}}_0(\Adversary)=1]-\Pr[\GAME^{\PVSS-\mathsf{IND}}_1(\Adversary)=1]| \leq \negl(\lambda)$, where  $\GAME^{\PVSS-\mathsf{IND}}_b(\Adversary)$ is in Figure \ref{figure-game-privacy}.

\begin{framedfigure}[Game $\GAME^{\PVSS-\mathsf{IND}}_b(\Adversary)$ with supporting interactive oracle $\OO_{\PVSS,\Adversary}(.)$ \label{figure-game-privacy}] 
\vspace{0.1cm}
\footnotesize \underline{$\GAME^{\PVSS-\mathsf{IND}}_b(\Adversary)$:}\\
\footnotesize $\pp\leftarrow\PVSS.\Setup(1^\lambda)$, $\CC\leftarrow\Adversary(\pp)$. If $|\CC|>t$ return $0$.\\
\footnotesize $((\pk_i,\sk_i),\pi^0_i) \leftarrow \PVSS.\Key\Gen(\pp)~\forall~i\not \in \CC$, \footnotesize $(\pk_i,\pi_i^0)_{i \in \CC}\leftarrow\Adversary(\pp,\CC)$,\\
\footnotesize Let $G=\{i\in [n]~|~\PVSS.\Key\Verify(\pp,\pk_i,\pi^0_i)=1\}$. \footnotesize Assume that $G=[n']$ and $[n] \setminus \CC \subseteq G$. If not, we re-enumerate the participants $P_i$ with $i \in G$ with an element in $[n']$.\\
\footnotesize $(s^0,s^1)\leftarrow \Adversary^{\OO_{\PVSS,\Adversary}(.)}(\pp,(\pk_i,\pi^0_i)_{i \in [n']})$.\\
\footnotesize \textcolor{gray}{Challenge phase.}\\
\footnotesize $(E^b,\pi^b) \leftarrow \PVSS.\Share(\pp,s^b,n',t)$.\\
\footnotesize $b' \leftarrow \Adversary((\pp,(\pk_i,\pi^0_i)_{i \in [n']},s^0,s^1,E^b,\pi^b)$.\\
\footnotesize Return $b'$.
\vspace{0.2cm}

\footnotesize \underline{Interactive oracle $\OO_{\PVSS,\Adversary}(s):$}\\
\footnotesize $(E=(E_i)_{i=1}^{n'},\pi^1)\leftarrow \PVSS.\Share(\pp,(\pk_i)_{i=1}^{n'},s,n',t)$.\\
\footnotesize $(s_i,\pi^2_i) \leftarrow \PVSS.\Dec(\pp,\pk_i,\sk_i,E_i)~\forall~i \not \in \CC,(s_i,\pi^2_i)_{i  \in \CC}\leftarrow \Adversary(\pp,(\pk_i,\pi^0_i)_{i \in [n']},E,\pi^1)$.\\
\footnotesize Let $S_2=\{i \in G~|~\PVSS.\Dec\Verify(\pp,E_i,s_i,\pi^2_i)=1\}$.\\
\footnotesize Return $\PVSS.\Combine(\pp,(s_i)_{i \in S})$.
\vspace{0.2cm}
\end{framedfigure}

\end{definition}

\subsection{Decentralized Random Number Generator}
We recall the definition of DRNG, which is adapted from \cite{MHN23} with minor modifications to capture that participants are given a common CRS.
\begin{definition}[DRNG, Adapted from \cite{MHN23}]
 A $(t,n)-$\textup{DRNG} on a set of participants $\mathcal{P}=\{P_1,P_2,\dots,P_n\}$ with output space $U$ is an epoch-based protocol, each epoch consists an algorithm $\DRNG.\Setup,$ two interactive protocols $\DRNG.\Init$, $\DRNG.\RandGen$ for participants in $\mathcal{P}$, and algorithms  $\DRNG.\Verify$  as follows:
\begin{enumerate}
\item $\crs \leftarrow \DRNG.\Setup(1^\lambda):$ On input a security parameter $1^\lambda$, this PPT algorithms returns a common reference string $\crs$.
\item $(\st,\qual,\pp,\{\sk_i\}_{i \in \qual})\leftarrow \DRNG.\Init(\crs)\interaction{\{P\}_{P \in \mathcal{P}}}$: This is an one-time protocol run by all participants in $\mathcal{P}$ given a common reference string $\crs$ to determine the list of qualified committees. At the end of the interaction, a set $\qual$ of qualified participants is determined, a global state $\st$ is initialized, and a list $\pp$ of public information is known to all participants. Each participant $P_i$ with $ i\in \qual$ also obtains his secret key $\sk_i$, only known to him.
\item $(\st:=\st',\qual,\Omega,\pi)\leftarrow \DRNG.\RandGen(\crs,\st,\pp)\interaction{\{P_i(\sk_i)\}_{i \in \qual}}:$ This is an interactive protocol between participants in a set $\qual$ each holding the secret key $\sk_i$ and common inputs $\st,\pp$. It is executed in each epoch. In the end, all honest participants output a value $\Omega \in U$ and a proof $\pi$ certifying the correctness of $\Omega$ made by the interaction. In addition, the global state $\st$ is updated into a new state $\st'$.
\item $b \leftarrow \DRNG.\Verify(\crs,\st,\Omega,\pi,\pp)$: This algorithm is run by a verifier. On input a common reference string $\crs$ (generated by a third party), the current state $\st$, a value $\Omega$, a proof $\pi$, a public parameter $\pp$ (generated by participants), this algorithm output a bit $b \in \{0,1\}$ certifying the correctness of $\Omega$.
\end{enumerate}
\end{definition}
As in \cite{MHN23} and previous works \cite{SJSW18,CSS19,GLOW20,DKMR21,BSKN23}, we require a DRNG to satisfy the four properties: \textit{pseudorandomness, Bias-resistance , liveness, public verifability}. 
\begin{definition}[Security of DRNG, Adapted from \cite{MHN23}]\label{definition-secure-drng} A \textup{secure} DRNG protocol is a DRNG protocol satisfying the following properties.
\begin{itemize}
\item \textbf{Pseudorandomness}. Let $\Omega_1,\Omega_2,\dots,\Omega_r$ be outputs generated so far. We say that a $(t,n)-$DRNG satisfies \textup{pseudo-randomness} if for any future outputs $\Omega_j$ where $j>r$ that has not been revealed, for any PPT adversaries $\Adversary$  who corrupts $t$ participants in $\mathcal{P}$, there exists a negligible function $\negl$ such that
\begin{equation*}
\left|\Pr\left[\Adversary(\Omega_j)=1\right]-\Pr\left[\Adversary(Y)=1\right]\right|\leq\negl(\lambda),
\end{equation*}
where $Y \uniformly U$ is an element chosen uniformly at random from the set $U$. 
\item \textbf{Bias-resistance.} For any adversary $\Adversary$ who corrupts $t$ participants, it cannot affect future random outputs for his own goal. 

\item \textbf{Liveness.} 
For any epoch, and for any adversary $\Adversary$ who corrupts  $t$ participants, the $\DRNG.\RandGen$ protocol is guaranteed to produce an output. 

\item \textbf{Public Verifiability.} Given  $\crs,\st',\pp,\Omega^*,\pi^* \in \{0,1\}^*$, an external verifier can run $\DRNG.\Verify(\crs,\st',\pp,\Omega^*,\pi^*)$ to determine the correctness of $\Omega^\star$.

\end{itemize}
\end{definition}

Several constructions, such as \cite{SJSW18,CSS19}, consider the weaker unpredictability property, which only requires that $\Adversary$ cannot correctly guess the future output. However, as pointed out by \cite{GLOW20}, pseudorandomness implies unpredictability. Thus, we will adapt the pseudorandomness property from \cite{MHN23}.

\section{The Underlying Lattice-Based PVSS}\label{section-the-underlying-pvss}
To construct a DRNG based on a PVSS in the standard model, we first need a lattice-based PVSS in the standard model. Hence, we briefly recall the PVSS construction in \cite{MNSN25}, which is our choice of PVSS for constructing the DRNG. 
\subsection{The PVSS Components}\label{section-pvss-components}
The PVSS in \cite{MNSN25} requires five components: A secret sharing scheme (SSS), a public key encryption scheme (PKE), and the NIZKs for key generation, sharing, and decryption. We briefly recall these components. For the SSS, the construction relies on Shamir SSS in Figure \ref{figure-sss}. For the PKE, the construction uses the ACPS PKE of \cite{ACPS09} in Figure \ref{figure-pke}, whose security relies on the LWE assumption in Section \ref{section-lattice}. In the PKE, each public key has a unique secret key, which is needed for the verifiability property (Definition \ref{definition-verifiability}).

\begin{framedfigure}[Shamir Secret Sharing Scheme \label{figure-sss}]
\vspace{-0.3cm}
\begin{itemize} 
\item $\SSS.\Share(\pp,s,n,t):$ Chooses a polynomla $p(X) \in \mathbb{Z}_p[X]$ of degree $t$. Compute $s_i=p(i) \pmod{p}$. Return $(s_i)_{i=1}^n$.
\item $\SSS.\Combine(\pp,S,(s_i)_{i \in S}):$ Compute $s=\sum_{i \in S} \lambda_{i,S}\cdot  s_i \pmod{p}$, where $\lambda_{i,S}=\prod_{j \in S,j \neq i}=j/(j-i) \pmod{p}$ are the Lagrange coefficients.
\end{itemize}
\end{framedfigure}

\begin{framedfigure}[The ACPS Public Key Encryption Scheme \label{figure-pke}]
\vspace{-0.3cm}
\begin{itemize}
\item $\PKE.\Setup(1^\lambda):$ Consider two positive integers $p,q$ with $q=p^2$ and $p$ is a prime number. Let $u,v,r$ be positive integers and $\alpha,\beta$  be positive real numbers.  Generate $\bA \uniform \mathbb{Z}_q^{v \times u}$, Return $(\bA,u,v,\alpha,\beta,r,p,q)$.
\item $\PKE.\Key\Gen(\bA):$ Sample $\bs \leftarrow D_{\mathbb{Z}^v,\alpha q},~\be \leftarrow D_{\mathbb{Z}^u,\alpha q}$. Repeat until $||\bs||<\sqrt{v} \cdot \alpha q$, $||\be||<\sqrt{u} \cdot \alpha q$. Compute $\bb=\bs^\top\cdot\bA+\be^\top \pmod{q}$. 
Return $(\pk,\sk)=(\bb,\bs)$.
\item $\PKE.\Enc(\bA,\bb,m):$ To encrypt a message $m \in \mathbb{Z}_p$, sample $\br \leftarrow D_{\mathbb{Z}^u,r}, e \leftarrow D_{\mathbb{Z},\beta q}$ and compute $\bc_1=\bA\cdot \br \pmod{q},~\bc_2=\bb\cdot\br+e+p\cdot m \pmod{q}$, where $\beta=\sqrt{u}\cdot \log u \cdot (\alpha+\frac{1}{2\cdot q})$. Return $(\bc_1,\bc_2)$.
\item $\PKE.\Dec(\bA,\bb,\bs, (\bc_1,\bc_2)):$ Compute $f=\bc_2-\bs^\top \bc_1 \pmod{p}$ and cast $f$ as an integer in $[-(p-1)/2,(p-1)/2]$. Return $m=(\bc_2-\bs^\top \bc_1-f)/p \pmod{p}$ with $f$ as the additional witness for decryption.
\end{itemize}
\end{framedfigure}

 Finally, we briefly recall the necessary NIZKs (Section \ref{section-nizk}) in \cite{MNSN25} to give a high-level overview of them. We need three NIZKs in total. In the PVSS of \cite{MNSN25}, these NIZKs are used to prove the following:
\begin{itemize}
\item NIZK for correct key generation. After generating the public key $\bb$ in Figure \ref{figure-pke}, the participant needs to convince the verifier that there exist secrets $\bs,\be$ corresponding to it. More specifically, given a statement $(\bA,\bb)$, it uses the NIZK to prove that there are witnesses $(\bs,\be)$ s.t. $\bb^\top=\bs^\top\cdot \bA+\be^\top \pmod{q}$ for some $\bs,\be$ having a small norm.

\item NIZK for correct sharing. According to \cite{MNSN25}, given that $(\bA,\bb_i)$ are valid public keys for all $i \leq n$, when sharing a secret, the dealer must prove that i) the ciphertexts are valid encryption and ii) the shares are valid shares of some secret $s$. In \cite{MNSN25}, we see that, for Shamir secret sharing scheme, condition ii) can be captured by using a parity check matrix $\bH^t_n \in \mathbb{Z}_p^{n \times (n-t-1)}$ s.t. $\bm=(m_1,\dots,m_n)$ is a valid share vector iff $\bm^\top\cdot \bH^t_n=\bzero \pmod{p}$. In the best case, the dealer, given the encryption and public keys $(\bA,(\bb_i,\bc_{1i},\bc_{2i})_{i=1}^n)$ would prove the existence of witnesses $(\br_i,e_i,m_i)_{i=1}^n$ st. i) $\bc_{1i}=\bA\cdot\br_i \pmod{q},\bc_{2i}=\bb_i\cdot \br_i+p\cdot m_i+e_i \pmod{q}$, $\br_i,e_i$ having a small norm for all $1\leq i \leq n$, and ii) $\bm^\top\cdot \bH^t_n=\bzero \pmod{p}$. In the worst case, then if a verifier accepts, it will be convinced that the value $m_i$ obtained by honestly decrypting $(\bc_{1i},\bc_{2i})$ using the secret key $\bs_i$ of $\bb_i$ must satisfy $\bm^\top\cdot \bH^t_n=\bzero \pmod{p}$, meaning that they are valid shares of some secret (such as $\bs_i$ must exist as $\bb_i$ is a valid public key). 
\item NIZK for correct decryption. Finally, each participant decrypts the shares and proves that the decryption result in Figure \ref{figure-pke} is correct.  More specifically, after obtaining $s_i$ from decryption, then from the statement $(\bA,\bb_i,(\bc_{1i},\bc_{2i}),s_i)$, participant $P_i$ needs to prove the existence of witnesses $(\bs_i,\be_i,f_i)$ s.t. $\bb_i=\bs_i^\top\cdot \bA+\be_i^\top \pmod{q}, \bc_{2i}-p\cdot s_i=\bs_i^\top\cdot \bc_{1i}+f_i \pmod{q}$ and $\bs_i,\be_i,f_i$ having small norm to convince the verifier that it has decrypted correctly.

\end{itemize}
 Due to limited space, we refer the formal NIZK constructions, and the security of the NIZKs to \cite{MNSN25}. The special property of the NIZKs is that they are proven secure in the standard model by combining \textit{trapdoor $\Sigma$-protocols} (see \cite{MNSN25}) and the compiler of \cite{LNPT20} instead of using the Fiat-Shamir paradigm. Note that the NIZKs require a third party to generate a CRS using the setup algorithm. However, the CRS only needs to be \textit{generated once}, then it can be used by participants to generate multiple random outputs (it remains the same even if participants are replaced). We believe this is acceptable, as remarked in Section \ref{section-our-contribution}. Finally, the NIZK compiler in \cite{LNPT20} has many complicated components that would make it hard to provide an exact cost. However, we can still estimate the cost of the NIZKs to be the cost of the trapdoor $\Sigma$-protocols, as analyzed in \cite{MNSN25}.

\subsection{The PVSS Construction}\label{section-pvss}
Now, we summarize the lattice-based PVSS of \cite{MNSN25} from the components above. It depends on the following primitives: i) the lattice-based public key encryption scheme $\PKE=(\PKE.\Setup,$ $\PKE.\Key\Gen,$  $\PKE.\Enc,\PKE.\Dec)$ described in Figure \ref{figure-pke}, ii) a Shamir secret sharing scheme $\SSS=(\SSS.\Share,\SSS.\Combine)$ described in Figure \ref{figure-sss} and iii) the three NIZKs $\NIZK_0,\NIZK_1,\NIZK_2$ for correct key generation, sharing, and decryption described in Section \ref{section-pvss-components}, where $\NIZK_i=(\NIZK_i.\Setup,$ $\NIZK_i.\Prove,\NIZK_i.\Verify)~\forall~0 \leq i \leq 2$ (their syntax is in Section \ref{section-nizk}).  The lattice-based PVSS construction in \cite{MNSN25} is briefly described in Figure \ref{figure-pvss}. The PVSS is proven secure in the standard model and achieves post-quantum security.

\begin{framedfigure}[The Lattice-based PVSS Construction in \cite{MNSN25} \label{figure-pvss}]
\begin{itemize}
\item $\PVSS.\Setup(1^\lambda):$ Generate $\bA \leftarrow \PKE.\Setup(1^\lambda)$ and common reference string $\crs_i$ for $\NIZK_i$ for all $i \in \{0,1,2\}$. Return $\pp=(\bA,(\crs_i)_{i=0}^2)$.
\item $\PVSS.\Key\Gen(\pp):$ Sample a public-secret key pair $(\bb,\bs) \leftarrow \PKE.\Key\Gen(\bA,\be)$ for some randomness $\be$ and provides a proof $\pi \leftarrow \NIZK_0\Prove(\crs_0,(\bA,\bb),(\bs,\be))$. Finally, return $((\bb,\bs),\pi)$.
\item $\PVSS.\Key\Verify(\pp,\bb,\pi):$ Return $\NIZK_0.\Verify(\crs_0,(\bA,\bb),\pi)$.
\item $\PVSS.\Share(\pp,(\bb_i)_{i=1}^n,s,n',t):$ The dealer enumerates all participants who passed the key verification by $\{1,2\dots,n'\}$. Calculate $(s_i)_{i=1}^{n'} \leftarrow \SSS.\Share(s,n',t)$. Then it computes $(\bc_{1i},\bc_{2i})\leftarrow \PKE.\Enc(\bA,\bb_i,s_i)$ for all $1 \leq i \leq n$ and keeps $(\br_i,e_i)$.  The proof then provides $\pi \leftarrow \NIZK_1.\Prove(\crs_1,(\bA',n',t,(\bb_i,\bc_{1i},\bc_{2i})_{i=1}^{n'}),(s_i,\br_i,e_i)_{i=1}^{n'})$. Returns $(E=(\bc_{1i},\bc_{2i})_{i=1}^{n'},\pi)$.
\item $\PVSS.\Share\Verify(\pp,(\bb_i)_{i=1}^{n'},n',t,E,\pi):$ To check the validity of the sharing, return $\NIZK_1.\Verify(\crs_1,(\bA',n',t,(\bb_i,\bc_{1i},\bc_{2i})_{i=1}^{n'}),\pi)$.
\item $\PVSS.\Dec(\pp,\bb_i,(\bc_{1i},\bc_{2i}),\bs_i):$ Compute $s_i=\PKE.\Dec(\bA,\bb_i,\bs_i,(\bc_{1i},\bc_{2i}))$ and receive additional witness $f_i$ for decryption. Then we provide a proof $\pi_i \leftarrow \NIZK_2.\Prove(\crs_2,(\bA,\bb_i,(\bc_{1i},\bc_{2i}),s_i),(\bs_i,\be_i,f_i))$. Return $(s_i,\pi_i)$.
\item $\PVSS.\Dec\Verify(\pp,\bb_i,\bc_{1i},\bc_{2i},s_i,\pi_i):$ Return $\NIZK_2.\Verify(\crs_2,(\bA,\bb_i,\bc_{1i},\bc_{2i}),s_i),\pi_i)$.
\item $\PVSS.\Combine(\pp,S,(s_i)_{i \in S}):$ If $|S| \leq t$ returns $\perp$. Otherwise, return $s=\SSS.\Combine(S,(s_i)_{i \in S})$.
\end{itemize}
\end{framedfigure}

\section{Latice-Based DRNG from PVSS in the Standard Model}\label{section-drng}
We construct a DRNG from the lattice-based PVSS from the previous section. We then sketch the security and analyze the complexity of the DRNG.
\subsection{Construction}
First, we describe a DRNG construction from any generic PVSS. The construction has been briefly described in Ouroboros \cite[Figure 12]{KRDO17} and used in \cite{CD17} for their DL-based PVSS. However, Ouroboros requires an RO outside the PVSS to generate the output. We \textit{modify} the DRNG so that we \textit{do not need to use any RO} and can instantiate the DRNG with any lattice-based PVSS.  Informally, the DRNG is as follows (assuming all participants follow the DRNG):
\begin{itemize}
\item Initially, for each $i$, each participant $P_j$ samples a public-s ecret key pair $(\pk_{ij},\sk_{ij})$ for $P_i$ to share his secret. Other participants then verify the validity of the keys and set $\qual$ to be participants who passed verification.
\item \textit{Re-enumerate} the set $\qual$ for participants 
who have passed key verification. For example, if $A,B,C$ have passed key verification, and initially they are enumerated with $2,3,5$ respectively ($\qual=\{2,3,5\}$). Then they will be re-enumerated with $1,2,3$, and $\qual$ will be updated as $\qual:=\{1,2,3\}$.
\item This is the start of an epoch. For each $i \in \qual$, participants samples $s_i \uniformly \mathbb{Z}_p$ and uses $\pk_i=(\pk_{i1},\dots,\pk_{in})$ to publish his transcript $(E_i,\pi^1_i)$. Other participants verify the validity of the transcript and denote $\qual'$ as the set of participants who published a valid sharing transcript.
\item For each $i \in \qual'$, participant $P_j$ with $j \in \qual'$ decrypts $E_{ij}$ to get the share $s_{ij}$ of $s_i$ and reveals it with its proof $\pi^2_{ij}$. Other participants then verify the decrypted shares $s_{ij}$ above.
\item If there are at least $t+1$ valid shares for $s_i$, it is recovered using Lagrange interpolation. Then the value  $\Omega_r=\sum_{i \in \qual'} s_i \pmod{p}$ is computed. The proof of the DRNG is all the public transcripts so far. Note that $\qual$ remains the same, and in the next epoch, those in $\qual$ are still allowed to share their secret (even if they are not in $\qual'$).
\end{itemize}

Unlike Ouroboros, our DRNG only relies on a $(n,t)$-PVSS $\PVSS=(\PVSS.\Setup,\PVSS.\Key\Gen \PVSS.\Key\Verify,$$\PVSS.\Share,\PVSS.\Share\Verify,\PVSS.\Dec,$ $\PVSS.\Dec\Verify,\PVSS.\Combine)$ does not require any  RO. The modified generic DRNG construction is formally described in Figure \ref{figure-drng}, assuming the existence of \textit{any PVSS}. In the description, we use $\crs$ to denote the parameter generated by $\PVSS.\Setup$ instead of $\pp$ like Figure \ref{figure-pvss}, because in the DRNG syntax, we need $\pp$ to denote the public keys $(\pk_i)_{i \in \qual}$ generated by participants themselves. Finally, we note that, whenever a new participant $P_j$ would like to join in a future epoch, then in $\DRNG.\Init$, we only need to provide a single tuple $\pk_j=(\pk_{j1},\dots,\pk_{jn})$ for $P_j$ to share its secret, and the tuples $\pk_i$ for old participants are the same as long as they are in $\qual$. \\

\begin{framedfigure}[The formal DRNG construction in each epoch \label{figure-drng}]
\begin{itemize}
\item $\DRNG.\Setup(1^\lambda):$ Generates $\crs \leftarrow \PVSS.\Setup(1^\lambda)$. Return $\crs$.
\item $\DRNG.\Init(\crs)\langle\{P\}_{P \in \mathcal{P}} \rangle:$  The participants proceed as follows:
\begin{enumerate}
\item For each $j \in [n]$, each participant $P_i$ generates $(\pk_{ji},\sk_{ji},\pi^0_{ji}) \leftarrow \PVSS.\Key\Gen(\crs)$. Participant $P_i$ then broadcasts $(\pk_{ji},\pi^0_{ji})_{j=1}^n$.
\item Other participants $P_j$ verify the validity of $\pk_{ji}$ by executing $b_{ji}=\PVSS.\Key\Verify(\crs,\pk_{ji},\pi^0_{ji})$.

\item Let $\qual=\{i \in [n]~|~b_{ji}=1~\forall~j \in [n]\}$. WLOG, $|\qual|=n'$. Re-enumerate $\qual$ to be $\qual:=\{1,2,\dots,n'\}$ (i.e., each qualified remaining participant will be re-enumerated with a value in $\{1,2,\dots,n'\}$). Also, for each $i \in \qual$, let $\pk_i=(\pk_{ij})_{j \in \qual}$ and $\sk_i=(\sk_{ij})_{j \in \qual}$. 
\item Return $(\st=\perp,\qual,\pp=(\pk_{i})_{i\in\qual},(\sk_{i})_{i \in \qual})$.

\end{enumerate}
\item $\DRNG.\RandGen(\crs,\st,\pp=(\pk_i)_{i \in \qual})\interaction{\{P_i(\sk_i)\}_{i \in \qual})}:$ To generate a random output $\Omega \in \mathbb{Z}_p$, the participants jointly proceed as follows:
\begin{enumerate}
\item Each participant $P_i$ with $i \in \qual$ generates a random secret $s_i \leftarrow \mathbb{Z}_p$ and computes $(E_i,\pi^1_i) \leftarrow \PVSS.\Share(\crs,\pk_i,s_i,n',t)$ where $\pk_i=(\pk_{ij})_{j \in \qual}$.
\item Other participants verify the validity of $(E_i,\pi^1_i)$ by executing $b_i=\PVSS.\Share\Verify(\crs,\pk_i,n',t,E_i,\pi^1_i)$. Let $\qual':=\{i \in \qual~|~b_i=1\}$.
\item For each $i \in \qual'$, other participants $P_j$ with $j \in \qual'$ reveals the share $s_{ij}$ of $s_i$ by computing $(s_{ij},\pi^2_{ij}) \leftarrow \PVSS.\Dec(\crs,\pk_{ij},\sk_{ij},E_{ij})$.
\item Other participants verify the  validity of $(s_{ij},\pi_{ij})$ by computing $b_{ij}=\PVSS.\Dec\Verify(\crs,\pk_{ij},E_{ij},\pi^2_{ij})$. 
\item Let $S_i=\{j \in \qual'~|~b_{ij}=1\}$. As soon as $|S_i|\geq t+1$, participants recover $s_i$ by computing $s_i=\PVSS.\Combine(\crs,S_i,(s_{ij})_{j \in S_i})$. 
\item The final random output is defined to be $\Omega=\sum_{i \in \qual'} s_i \pmod{p}$. If some $s_i$ cannot be reconstructed, then $\Omega=\perp$.
\item The proof $\pi$ for the DRNG is simply the public transcript so far. Thus $\pi=(\qual',(\pi^0_{ij})_{i,j\in \qual'},(E_i,\pi^1_i)_{i \in \qual'},(s_{ij},\pi^2_{ij})_{i\in\qual',j \in S_i})$. 
\item Return $(\st=\perp,\Omega,\pi)$. Note that $\qual$ remains the same for all epochs, and in the next epoch, all participants in $\qual$ (including those not in $\qual'$ in the current epoch) are still allowed to share their secret like Step $1$.
\end{enumerate} 
\item $\DRNG.\Verify(\crs,\st,\Omega,\pi,\pp= (\pk_i)_{i \in \qual}):$ A public verifier proceeds as follows:
\begin{enumerate}
\item For each $i \in \qual$, check if  $\PVSS.\Key\Verify(\crs,\pk_{ij},\pi^0_{ij})=1$ for all $j \in \qual$.
\item For each $i \in \qual'$, check if $\PVSS.\Share\Verify(\crs,\pk_i,E_i,\pi^1_i)=1$.
\item For each $i \in \qual', j \in S_i$, check if $\PVSS.\Dec\Verify(\crs,\pk_{ij},E_{ij},\pi^2_{ij})=1$.
\item Compute $s_i=\PVSS.\Combine(\crs,S_i,(s_{ij})_{j \in S_i})$.
\item Finally, check if $\Omega=\sum_{i \in \qual'} s_i \pmod{p}$. Accept iff all checks pass.
\end{enumerate}
\end{itemize}
\end{framedfigure}

\noindent \textbf{Instantiation.} The compiler only requires a generic PVSS. Thus, to achieve post-quantum security and in the standard model, we propose to use the PVSS described in Section \ref{section-pvss} instead of DL-based PVSSs such as SCRAPE or Ouroboros. Our DRNG is proven secure in the standard model, assuming that participants have access to a CRS generated by 
$\PVSS.\Setup$. However, this CRS only needs to be generated once, as discussed earlier.

\subsection{Security Proof}
Here, we sketch the security proof of the DRNG. While Ouroboros provided a DRNG compiler using a generic PVSS, their compiler requires ROM and only proved the security when it is instantiated with a DL-based PVSS. Therefore, we need to show that our modified DRNG achieves all the required security when it is instantiated with a generic PVSS (consequently, we have a post-quantum secure DRNG in the standard model using the PVSS in Section \ref{section-pvss}).  
\begin{theorem}
The DRNG in Figure \ref{figure-drng} satisfies the pseudorandomness property.
\end{theorem}
\begin{proof}
To show pseudorandomness, we need to show that, for any PPT adversary $\Adversary$, \textit{before the reveal phase} (Step $3$ of $\DRNG.\RandGen$), then $\Adversary$ cannot find any pattern to distinguish between the correct DRNG output given to him and a truly random output in $\mathbb{Z}_p$. To capture this setting, we consider the pseudorandomness security game $\GAME_b^{\mathsf{Psd-DRNG}}(\Adversary)$ for the DRNG in Figure \ref{figure-game-drng}.
\begin{framedfigure}[The game $\GAME_b^{\mathsf{Psd-DRNG}}(\Adversary)$ \label{figure-game-drng}] 
\begin{enumerate}
\vspace{-0.35cm}
\item The challenger computes the CRS $\crs$ from $\PVSS.\Setup$ and gives $\crs$ to $\Adversary$.
\item The adversary $\Adversary$ chooses a set of corrupted participants $\CC$ s.t. $|\CC| \leq t$. Note that $[n] \setminus \CC \subseteq \qual$.
\item The challenger executes $((\pk_{ji},\sk_{ji}),\pi^0_{ji})\leftarrow \PVSS.\Key\Gen(\crs)$ for all $i \not \in \CC, j \in [n]$ , which represents the public keys of honest participants. The adversary $\Adversary$ also sends the public keys (possibly invalid or might not send at all) of corrupted participants $(\pk_{ji},\pi^0_{ji})_{j \in [n],i \in \CC}$ as well.
\item The challenger verifies $\pk_{ij}$ by executing $\PVSS.\Key\Verify(\crs,\pk_{ji},\pi^0_{ji})$ for all $j \in [n], i\in \CC$ and excludes all participants in $\CC$ with invalid public keys. Also, it re-enumerates $\qual:=[n']$ and $\CC$ after this. 
\item For $i=1,2,\dots,r-1$, the challenger and $\Adversary$ jointly executes the DRNG to generates $\Omega_1,\Omega_2,\dots,\Omega_{r-1}$ using $\DRNG.\RandGen$ in Figure \ref{figure-drng}.
\item Eventually, in the $r$-th epoch, for each honest $P_i$ with $i \in \qual$, the challenger samples $s_i \uniformly \mathbb{Z}_p$ computes $(E_i,\pi^1_i) \leftarrow \PVSS.\Share(\crs,\pk_i,s_i,n',t)$ and broadcasts $(E_i,\pi^1_i)$. It also receives the pairs $(E_i,\pi^1_i)_{i  \in\qual \cap \CC}$ from $\Adversary$.
\item The challenger verifies the validity of the transcripts $(E_i,\pi^1_i)_{i \in\qual\cap \CC}$ by executing $\PVSS.\Share\Verify(\crs,\pk_i,n',t,E_i,\pi^1_i)$ for all $i \in \qual\cap\CC$ and determines the set $\qual'$ as in Figure \ref{figure-drng} (which includes all honest participants due to the correctness property of the PVSS).
\item For each $i \in \qual' \cap \CC$, the challenger uses $\sk_{ij}$ to restore the share $s_{ij}$ by computing $(s_{ij},.)=\PVSS.\Dec(\crs,\pk_{ij},\sk_{ij},E_{ij})$ for all $j \in \mathcal{G}$ and uses Lagrange interpolation to recover $s_i$. Note that $|\mathcal{G}| \geq t+1$, thus due to the correctness and verifiability of the PVSS, the restored result $s_i$ is the same restored secret of $P_i$ in the real execution of the DRNG except with negligible probability. 
\item Now, suppose $n' \in \qual'$ and $n' \not \in \CC$ (i.e., $P_{n'}$ is honest). Then the challenger computes $A^\star=\sum_{i \in \qual,i \neq n'} s_i \pmod{p}$. If $b=0$, then $\Omega_r=A+s_{n'} \pmod{p}$, otherwise, the challenger samples $s' \uniformly \mathbb{Z}_p$ and compute $\Omega_r=A+s' \pmod{p}$. The challenger then returns $\Omega_r$ to $\Adversary$.
\item $\Adversary$ outputs a bit $b'$, which is the result of the experiment.
\end{enumerate}
\end{framedfigure}

We see that, according to Definition \ref{definition-secure-drng}, pseudorandomnes holds if the advantage $\Adv^{\mathsf{Psd-DRNG}}(\Adversary)=|\Pr[\GAME_0^{\mathsf{Psd-DRNG}}(\Adversary)=1]-\Pr[\GAME_1^{\mathsf{Psd-DRNG}}(\Adversary)=1]|$ is negligible. Indeed, up to Step $4$ of  Figure \ref{figure-game-drng}, $\Adversary$ has received the outputs $\Omega_1,\dots,\Omega_{r-1}$ of the DRNG. In the $r$-th epoch (starting from Step $5$), when $b=0$, $\Adversary$ is given the \textit{correct} future output $\Omega_r=\sum_{i \in \qual'} s_i \pmod{p}$ of the DRNG (except with negligible probability). Otherwise, $\Adversary$ is given a \textit{random} output  in $\mathbb{Z}_p$ as $s'$ is random in $\mathbb{Z}_p$. Thus we just need to prove that $\Adv^{\mathsf{Psd-DRNG}}(\Adversary) \leq \negl(\lambda)$ for some negligible function $\negl$. Suppose otherwise, we prove that there is an adversary $\Adversary'$ that uses $\Adversary$ to break the IND2-privacy property of the PVSS in Figure \ref{figure-game-privacy}. The algorithm $\Adversary'$ is described in Figure \ref{figure-adversary}.
\begin{framedfigure}[The reduction algorithm $\Adversary'$ \label{figure-adversary}]
\vspace{-0.35cm}
\begin{enumerate}
\item First, $\Adversary'$ receives $\crs$ from the IND2-privacy challenger  (who acts on behalf of honest participants in the PVSS). It provides $\crs$ to $\Adversary$ and receices the set $\CC$ from $\Adversary$. The adversary $\Adversary'$ then submits $\CC$ to the IND2-privacy challenger.
\item  Let $\mathcal{G}$ denote the set of honest participants ($\Adversary'$ would ''act`` on their behalf when interacting with $\Adversary$) and consider an honest participant $P$. WLOG, initially, $P$ is enumerated as $n$. $\Adversary'$ honestly samples $((\pk_{ji},\sk_{ji}),\pi^0_{ji})\leftarrow \PVSS.\Key\Gen(\crs)$ for all $j \in [n] \setminus \{n\},i \in \mathcal{G}$. It also receives $(\pk_{ni},\pi^0_{ni})_{i \in \mathcal{G}}$ from the IND2-privacy challenger, which it will use as the public key for $P$.

\item $\Adversary'$ forwards $(\pk_{ji},\pi^0_{ji})_{j\in [n],i \in \mathcal{G}}$ to $\Adversary$, while also receives $(\pk_{ji},\pi^0_{ji})_{j\in[n],i \in\CC}$ from $\Adversary$. $\Adversary'$ then just verify the keys of $\Adversary$ and re-enumerates $\qual, \mathcal{G}$. Suppose $\qual=[n']$ and $P$ is re-enumerated as $n'$.
\item For $1\leq i\leq r-1$, $\Adversary'$ acts on the behalf of honest participants in the $i$-th epoch as follows:
\begin{enumerate}
\item For each $k \in \mathcal{G},k\neq n'$, it samples a secret $s'_k \uniformly \mathbb{Z}_p$ for $P_k$, then uses the PVSS with public keys $\pk_k=(\pk_{kj})_{j \in \qual}$  to share $s'_k$. 
\item For $P$, it queries random secret $s'_{n'} \uniformly \mathbb{Z}_p$ to the IND2-privacy challenger, receives the sharing transcript $(E_{n'},\pi^1_{n'})$ from $\OO_{\PVSS,\Adversary}(s'_{n'})$ (see Figure \ref{figure-game-privacy}), and then forwards it to $\Adversary$. $\Adversary'$ then verify the transcript of participants in $\qual\cap \CC$ and determines $\qual'$ for the $i$-th epoch.
\item For each $k \in \mathcal{G},k\neq n'$,  $\Adversary'$ simply reconstruct $s'_k$ it together with $\Adversary$ using the secret keys $(\sk_{kj})_{j \in \mathcal{G}}$ to decrypt the shares $(s'_{kj})_{j \in \mathcal{G}}$.
\item When needing to reconstruct $s'_{n'}$, it receives the honest shares $(s'_{n'j},\pi^2_{n'j})_{j \in \mathcal{G}}$ of $s'_{n'}$ from the IND2-privacy challenger from $\OO_{\PVSS,\Adversary}(s'_{n'})$ and forwards them to $\Adversary$. It also receives the decrypted shares and proofs $(s'_{n'j},\pi^2_{n'j})_{j \in \qual'\cap \CC}$ from $\Adversary$ and forward this to the challenger to complete the query $\OO_{\PVSS,\Adversary}(s'_{n'})$. 
\item Finally, with the secrets $(s'_k)_{k \in \qual'}$ are revealed, the value $\Omega_i$ is computed as $\Omega_i=\sum_{i \in \qual'} s'_i \pmod{p}$.
\end{enumerate}
\item  In the $r$-th epoch, for each $i \in \mathcal{G} \setminus \{n'\}$, $\Adversary'$ samples $s_i \leftarrow \mathbb{Z}_p$ and computes $(E_i,\pi^1_i) \leftarrow \PVSS.\Share(\crs,\pk_i,s_i,n',t)$. For $n'$, $\Adversary'$ samples $s^0_{n'} \leftarrow \mathbb{Z}_p, s^1_{n'} \leftarrow \mathbb{Z}_p$ and sends $s^0_{n'},s^1_{n'}$ to the challenger (it begins the challenge phase in Figure \ref{figure-game-privacy}). It receives $(E_{n'},\pi^1_{n'})$ which is the sharing transcript of $s^0_{n'}$ or $s^1_{n'}$. $\Adversary'$ forwards $(E_i,\pi^1_i)_{i \in \mathcal{G}}$ to $\Adversary$, while also receiving $(E_i,\pi^1_i)_{i \in \qual\cap \CC}$ from $\Adversary$. 
\item $\Adversary'$ verifies all the transcripts $(E_i,\pi^1_i)_{i \in \CC}$ and determines the set $\qual'$. For each $i \in \qual' \cap \CC$, $\Adversary'$ uses $\sk_{ij}$ to restore the share $s_{ij}$ for all $j \in \mathcal{G}$ and uses Lagrange interpolation to recover $s_i$. $\Adversary'$ then computes $\Omega_r=s^0_{n'}+\sum_{i \in \qual',i \neq n'} s_i \pmod{p}$ and sends $\Omega_r$ to $\Adversary$.
\item Finally, $\Adversary'$ outputs whatever $\Adversary$ outputs.
\end{enumerate}
\end{framedfigure}

In Figure \ref{figure-adversary}, when $(E_{n'},\pi^1_{n'})$ is the sharing transcript of $s^0_{n'}$, then $\Adversary'$ interacts with $\GAME^{\mathsf{IND-PVSS}}_0$ (Figure \ref{figure-game-privacy}) and has produced the transcript of $\GAME^{\mathsf{Psd-DRNG}}_0$ to $\Adversary$. Otherwise, we have $(E_{n'},\pi^1_{n'})$ is the sharing transcript of $s^1_{n'}$, in this case $s^0_{n'}$ plays the role of $s'$ in Figure \ref{figure-game-drng}, thus $\Adversary'$ interacts with $\GAME^{\mathsf{IND-PVSS}}_1$ and has produced the transcript of $\GAME^{\mathsf{Psd-DRNG}}_1$ to $\Adversary$. Hence, $\Pr[\GAME_b^{\mathsf{Psd-DRNG}}(\Adversary)=1]=\Pr[\GAME_b^{\mathsf{IND-PVSS}}(\Adversary')=1]$ for $b \in \{0,1\}$. Thus, $\Adv^{\mathsf{Psd-DRNG}}(\Adversary) \leq \Adv^{\mathsf{IND-PVSS}}(\Adversary')$. So if $\Adv^{\mathsf{Psd-DRNG}}(\Adversary)$ is non negligible, then $\Adv^{\mathsf{IND-PVSS}}(\Adversary')$ is non-negligible, contradiction due to the IND2-privacy of the PVSS (Definition \ref{definition-pvss-privacy}), which states that $\Adv^{\mathsf{IND-PVSS}}(\Adversary')$ must be negligible. Thus, the DRNG achieves pseudorandomness.  \qed
\end{proof}
\begin{theorem}
The DRNG in Figure \ref{figure-drng} satisfies the availability property.
\end{theorem}
 \begin{proof}
 Let $\qual'$ denote the set of participants passed verification after Step $3$ of $\DRNG.\RandGen$. Due to the verifiability of the PVSS, for each $i \in \qual'$, all participants would agree on some $s_i$ from $(E_i,\pi^1_i)$. It suffices to prove that all $s_i$ are reconstructed for all $i \in \qual'$ from $(E_i,\pi^1_i)$. For each $i \in \qual'$, the secrets $s_i$ can be recovered from  $t+1$ correct shares $s_{ij}$ due to the correctness and verifiability of the PVSS (those $s_{ij}$ that makes $\PVSS.\Dec\Verify$ returns $1$). Since there are $n-t \geq t+1$ honest participants, there are at least $t+1$ correct shares, so $s_i$ can be reconstructed for all $i \in \qual'$ using $\PVSS.\Combine$.  \qed
 \end{proof}
 \begin{theorem}
The DRNG in Figure \ref{figure-drng} satisfies the bias-resistance  property.
\end{theorem}
\begin{proof}
It is implied by pseudorandomness and availability. Indeed, if $\Adversary$ wishes to affect the output, it must do so during the sharing or reconstruction phase. Note that due to the pseudorandomness property, during the sharing phase, an adversary controlling $t$ participants cannot find any pattern to distinguish between the DRNG output $\Omega=\sum_{i \in \qual'} s_i \pmod{p}$ and a truly random output. During the reconstruction phase, for each $i \in \qual'$, with at least $t+1$ honest participants, the same result $s_i$ will be restored successfully (we have proved this in the availability property). Thus, the same value $\Omega$ is restored and cannot be changed, so $\Adversary$ cannot affect the output in this phase. Hence, the adversary cannot bias the output.  \qed
\end{proof}
 \begin{theorem}
The DRNG in Figure \ref{figure-drng} satisfies the public verifiability property.
\end{theorem}
\begin{proof}
The verifier can simply execute $\DRNG.\Verify$ to check the DRNG. First, it uses $\PVSS.\Key\Verify$ to check the correctness of the keys $\pk_i$, then it use $\PVSS.\Share\Verify$ to check the correctness of $(E_i,\pi^1_i)$ for all $i \in \qual'$. Finally, it uses $\PVSS.\Dec\Verify$ to check the correctness of the decrypted shares $s_{ij}$ for all $i \in \qual', j \in S_i$. For honest $P_i$, we easily see that the correctness of the PVSS (Definition \ref{def-correctness}) implies that if the secret $s_i$ is shared, it will be restored, and the verifier accepts. Conversely, even for honest participants, then verifiability (Definition \ref{definition-verifiability}) implies that if the verifier accepts, then the sharing transcript $(E_i,\pi^1_i)$ must be valid transcript of some secret $s_i$ and the decrypted shares $(s_{ij})_{j \in S_i}$ are valid shares of $s_i$. So the verifier is convinced that all participants in $\qual'$ have executed the PVSS correctly to share some secret $s_i$, which implies the correctness of $\Omega_r$ as desired.  \qed
\end{proof}

 Finally, we do not have to worry about $\Adversary$ knowing the output $\Omega$ much sooner than honest participants like those using VDFs because the algorithms $\PVSS.\Dec\Verify$, $\PVSS.\Combine$ only take polynomial time in $n,\lambda$ (see Section \ref{seciton-complexity-analysis}) to compute (not too long, hence honest participants can agree on the outputs not much longer than those with improved computational power), while it takes a lot of time (e.g., a day or several hours) to compute a VDF output.
 
\subsection{Complexity Analysis}\label{seciton-complexity-analysis}
We analyze the complexity of our DRNG, which includes the costs of $\DRNG.\Init$ and $\DRNG.\RandGen$. We include the cost of $\DRNG.\Init$ because it will be re-executed every time a participant is replaced/changed before $\DRNG.\RandGen$. The cost of $\DRNG.\Verify$ is the same as $\DRNG.\RandGen$. Note that there are two factors affecting the total time of an epoch: The computation complexity per participant and the number of rounds. One is the number of required steps for a participant to perform necessary computations (executing algorithms), while the other is the number of times required to exchange messages online (which is affected by the network delay). In \cite{MNSN25}, to construct the NIZK for a language, the authors designed the trapdoor $\Sigma$-protocol for the language first, then used the compiler of \cite{LNPT20} to achieve the NIZK. While the NIZK compiler in the standard model requires many complicated components, it would be hard to give a concrete complexity. However, as in \cite{MNSN25}, it is possible to analyze the cost of the trapdoor $\Sigma$-protocols and estimate the cost of the NIZK to be at least the trapdoor $\Sigma$-protocols, and consequently, it is possible to estimate the cost of the PVSSs as well. We will use the notation $\Omega$ to imply that the cost will be more than estimated. As in \cite{MNSN25}, we denote $v,u$ as the lattice parameters used in the PKE and $\lambda$ as the security parameter (recall Figure \ref{figure-pke}).\\ \vspace{-0.27cm}

\noindent \textbf{Communication Complexity.} In $\DRNG.\Init$, each $P_i$ has to submit his public key $\pk_{ji}$ and proof $\pi_{ji}$ using $\PVSS.\Key\Gen$. Since the total $n^2$ instance of $\PVSS.\Key\Gen$ is executed, it incurs $\Omega(n^2(u+v) \log q)$ cost. In $\DRNG.\RandGen$, when sharing the secret, the dominating computation complexity for each dealer is $\Omega(n(u+v) \log q)$ (analyzed in \cite{MNSN25}). Thus, the communication complexity in this step is $\Omega(n^2(u+v) \log q)$. Finally, each participant needs to broadcast their decrypted shares for each secret using $\PVSS.\Dec$. According to \cite{MNSN25}, the dominating cost for this is $\Omega(n^2(u+v) \log q)$ for verifying $n^2$ shares. In conclusion, the total communication complexity is estimated to be $\Omega(n^2(u+v) \log q)$. \\ \vspace{-0.27cm}

\noindent \textbf{Computation Complexity.} In $\DRNG.\Init$, each participant computes his public key $\pk_{ji}$ and proof $\pi_{ji}$ for all $j \in [n]$. It then verifies the correctness of the keys. According to \cite{MNSN25}, computing and verifying $O(n^2)$ public keys would require $\Omega(n^2uv)$ cost. In $\DRNG.\RandGen$, each participant uses $\PVSS.\Share$ to compute $(E_i,\pi^1_i)$ and then verify other transcripts. According to \cite{MNSN25}, this would require at least $\Omega(\lambda(n^3+n^2 uv))$ complexity to compute and verify $n$ transcripts. Each participant then decrypts $O(n)$ shares and verifies $O(n^2)$ decrypted shares. The complexity for this would be $\Omega(n^2uv)$. Finally, reconstructing the secret takes $O(n^2 \log^2 n)$ cost. In conclusion, the computation complexity is estimated to be  $\Omega(\lambda(n^3+n^2 uv))$. The verification cost for an external verifier is also the same. \\ \vspace{-0.27cm}

\noindent \textbf{Round Complexity.} To generate a single random value, we require two rounds in $\DRNG.\RandGen$. The first round consists of Step $1$ and Step $2$, while the second round consists of the remaining steps. In the first round, participants only need to send  $(E_i,\pi^1_i)$ in Step $1$ to other participants, and they can compute $\PVSS.\Share\Verify$ and determine $\qual'$ by themselves. In the second round, participant $P_i$ only needs to send his decrypted shares $(s_{ji},\pi^2_{ji})_{j \in \qual'}$ in Step $3$, then other participants (or anyone) can compute $\PVSS.\Dec\Verify$, $S_i,\Omega,\pi$ by themselves without needing to send any more messages. Here, 
 note that previous works in cryptography (e.g., \cite{GLOW20,CGGSP21,AJMMST23,EKT24,K24}, where non-interactive means one round) determine the number of rounds in the same way. When the information on the broadcast channel is sufficient to publicly determine the output ($\Omega$ in our case), then everyone (including an external verifier) can determine it without requiring more messages from participants. In \textit{practice}, to let everyone see and agree on $\Omega$ intermediately, then even if $\Omega$ can be computed independently and publicly, we 
 still ask participants to \textit{broadcast their result}. The result broadcast by most participants will be $\Omega$. However, like the work above, this broadcast process is more of an implementation detail and thus is not treated as a round.

\section{Conclusion}\label{section-conclusion}
In this paper, we proposed the first post-quantum DRNG from a lattice-based PVSS that can achieve security in the standard model and requires only two rounds of communication. Although the DRNG requires that participants have access to a CRS, this string needs to be generated once, and participants can use this string to generate many random numbers. Hence, we only need to place trust in a third party once. In the DRNG, we employ the technique of \cite{MNSN25}, which only allows us to estimate the cost of the DRNG so far, rather than a more concrete cost. Hence, it would be helpful if there were any optimization that allows us to achieve security in the standard model with reduced and more concrete cost. We would like to leave the above-mentioned issue for future research.

\vspace{-0.15cm}
\begin{spacing}{0.9}
\bibliographystyle{abbrv}
\bibliography{refs.bib}
\end{spacing}

\end{document}

%% file: macros.tex
\newcommand{\crs}{\mathsf{crs}}
\newcommand{\LL}{\mathcal{L}}

\newcommand{\OO}{\mathcal{O}}
\newcommand{\uniform}{\xleftarrow{\$}
}
\newcommand{\bH}{\mathbf{H}}

\newcommand{\GAME}{\mathbf{Game}}

\newcommand{\PVSS}{\mathsf{PVSS}}
\newcommand{\Dec}{\mathsf{Dec}}
\newcommand{\Key}{\mathsf{Key}}
\newcommand{\Gen}{\mathsf{Gen}}

\newcommand{\xmark}{\ding{55}}%

\newcommand{\Adversary}{\mathcal{A}}

\newcommand{\PKE}{\mathsf{PKE}}

\newcommand{\interaction}[1]{\left\langle #1\right\rangle}

\newcommand{\negl}{\mathsf{negl}}

\newcommand{\pk}{\mathsf{pk}}

\newcommand{\bx}{\mathbf{x}}
\newcommand{\bb}{\mathbf{b}}

\newcommand{\CRS}{\mathcal{CRS}}

\newcommand{\sk}{\mathsf{sk}}

\newcommand{\tempcaption}{}

\newcommand{\condprob}[2]{\Pr\left[ \begin{array}{l}
		#1
	\end{array}\;\middle|\;\begin{array}{l}
		#2
	\end{array} \right]}

\newcommand{\Share}{\mathsf{Share}}

\newcommand{\st}{\mathsf{st}}

\newcommand{\DRNG}{\mathsf{DRNG}}
\newcommand{\Init}{\mathsf{Init}}
\newcommand{\RandGen}{\mathsf{RandGen}}

\newcommand{\qual}{\mathsf{QUAL}}

\newcommand{\etal}{\textit{et al.}}
\newcommand{\pp}{\mathsf{pp}}

\newcommand{\by}{\mathbf{y}}

\newcommand{\bm}{\mathbf{m}}

 \newcommand{\bs}{\mathbf{s}}

 \newcommand{\be}{\mathbf{e}}
  \newcommand{\bA}{\mathbf{A}}
    
        \newcommand{\bzero}{\mathbf{0}}
   \newcommand{\RR}{\mathcal{R}}

 \newcommand{\uniformly}{\stackrel{\$}{\leftarrow}}
  
  \newcommand{\br}{\mathbf{r}}

  \newcommand{\Enc}{\mathsf{Enc}}

\newcommand{\SSS}{\mathsf{SSS}}

\newcommand{\Prove}{\mathsf{Prove}} 
 
\newcommand{\Verify}{\mathsf{Ver}}
 
\newcommand{\NIZK}{\mathsf{NIZK}} 
\newcommand{\Setup}{\mathsf{Setup}} 
 
\newcommand{\Combine}{\mathsf{Combine}} 
 
\newcommand{\bc}{\mathbf{c}} 
\newcommand{\Adv}{\mathbf{Adv}}

\newcommand{\CC}{\mathcal{C}}
\newcommand{\Correctness}{\mathsf{Correctness}}

\newenvironment{framedfigure}[1][\makered{Missing caption}]
{
	\def\myenvargumentI{#1} 
	\begin{figure*}[!h]
    
		      \begin{tabular}{|p{\textwidth}|}

			\hline 
            \footnotesize
		}
		{ 
        \footnotesize
			\\\hline
		\end{tabular}
		\caption{\myenvargumentI}
		\vspace*{-1em}
	\end{figure*}
}

{\begin{enumerate}}%
{\end{enumerate}}%